\documentclass[pra,aps,superscriptaddress,twocolumn,,showpacs,10pt]{revtex4-1}

\usepackage{graphicx,epsfig,amsmath,latexsym,amssymb,verbatim,color,bm}

\usepackage{theorem}
\newtheorem{definition}{Definition}
\newtheorem{lemma}[definition]{Lemma}

\def\squareforqed{\hbox{\rlap{$\sqcap$}$\sqcup$}}
\def\qed{\ifmmode\squareforqed\else{\unskip\nobreak\hfil
\penalty50\hskip1em\null\nobreak\hfil\squareforqed
\parfillskip=0pt\finalhyphendemerits=0\endgraf}\fi}
\def\endenv{\ifmmode\;\else{\unskip\nobreak\hfil
\penalty50\hskip1em\null\nobreak\hfil\;
\parfillskip=0pt\finalhyphendemerits=0\endgraf}\fi}
\newenvironment{proof}{\noindent \textbf{{Proof~} }}{\qed}
\newcommand{\be}{\begin{eqnarray}}
\newcommand{\ee}{\end{eqnarray}}
\newcommand{\dpsi}{\mathrm{d}\Psi}
\newcommand{\tr}{\mathrm{tr}}
\usepackage{graphicx}
\begin{document}

\title{Activation of entanglement in teleportation}
%\title{Noisy classical channel in the teleportation protocol for qudits}
%\title{Minimal capacity of classical channel in teleportation protocol}
\author{Ryszard Weinar}
\affiliation{Institute of Theoretical Physics and Astrophysics, University of Gda\'nsk, 80-952 Gda\'nsk, Poland}
\author{Wies\l{}aw Laskowski}
\affiliation{Institute of Theoretical Physics and Astrophysics, University of Gda\'nsk, 80-952 Gda\'nsk, Poland}
\author{Marcin Paw\l{}owski}
\affiliation{Institute of Theoretical Physics and Astrophysics, University of Gda\'nsk, 80-952 Gda\'nsk, Poland}
\affiliation{Department of Mathematics, University of Bristol, Bristol BS8 1TW, U.K.}

\begin{abstract}
We study the activation of entanglement in teleportation protocols. To this end, we a present derivation of the average fidelity of teleportation process with noisy classical channel for qudits. In our work we do not make any assumptions about the entangled states shared by communicating parties.  Our result allows us to specify the minimum amount of classical information required to beat the classical limit when the protocol is based on the Bell measurements. We also compare average fidelity of teleportation obtained using noisy and perfect classical channel with restricted capacity. The most important insight into the intricacies of quantum information theory that we gain is that though entanglement, obviously, is a necessary resource for efficient teleportation it requires a certain threshold amount of classical communication to be more useful than classical communication. Another interesting finding is that the amount of classical communication required to activate entanglement for teleportation purposes depends on the dimension $d$ of the system being teleported but is not monotonic reaching maximum for $d=4$.
\end{abstract}

\pacs{03.67.Hk}

\maketitle
\section{introduction}
The teleportation protocol is a widely used and tested tool \cite{Bennett93,Dur00,Vaidman94,Horodeccyreview09}. It became main part of many quantum communication protocols and still is an interesting research field \cite{Jung09,DiFranco12,Taketami11}. It allows to transmit an unknown quantum state from a sender traditionally named "Alice" to a receiver "Bob" who are spatially separated. This protocol consists of a classical and a quantum channel. The presence of noise in these channels introduces imperfections in the process. A popular way to describe efficiency of teleportation is through the average fidelity \cite{Nielsen10}. When the protocol works perfectly, fidelity is equal to 1, which is the maximum value.  Otherwise it is less. A lot of works has been devoted to this subject \cite{Horodeccy99,Oh02}, but very little is said in them about the efficiency of teleportation, depending only on the classical channel's noise. While we have perfect classical communication we do not have to care about it and this is what people usually do. Things change when we are interested in the amount of information needed to have a non classical process. In one of the recent papers \cite{Banik12} authors  considered a teleportation protocol of qubits with imperfect classical channel where no restrictions on the quantum channel were made. They also derived a minimal capacity of a classical channel needed to have fidelity grater than the classical maximum of $\frac{2}{3}$.

The motivation for this work is to determine the minimum capacity of a classical channel in a more general case. We ask what is the minimum amount of information sent via classical channel if the average fidelity is to exceed the classical limit in case of qudits. This provides us with a threshold value of classical communication. Without any amount of entanglement is, for teleportation purposes, less useful than classical communication. Another interesting finding is that this threshold value for entanglement activation depends on the dimension $d$ of the system being teleported but is not monotonic reaching maximum for $d=4$.

%In this work, we restrict ours studies to the $d$ dimensional protocol based on the Bell measurements.  The answer we derived shows that the minimum capacity of a noisy channel depending on the dimension of the transferred state is not monotonic function and has maximum for $d=4$. This function for $d=2$ corresponds to the result of \cite{Banik12}.
%

In the first part of this article we give a description of the teleportation protocol and the notation we use. We also derive the value of the average fidelity of the teleportation process in the case where perfect channel is replaced with a noisy one. Later we give formula, based on the Shannon entropy, for a minimum capacity of the noisy classical channel and using it we move to our main result. The last part is devoted to the comparison of two channels: perfect with limited capacity and noisy.

\section{TELEPORTATION PROTOCOL}
Here we briefly present the teleportation protocol for qudits which we are operating with to establish the notation.

Alice has to teleport a qudit $|\Psi\rangle=\sum_{i=0}^{d-1}a_i|i\rangle$. Both Alice and Bob share the maximally entangled state  $|\Phi\rangle=\frac{1}{\sqrt d}\sum_{i=0}^{d-1}|i\rangle\otimes|i\rangle$. To start the teleportation Alice measures the states (two particles) possessed by her in Bell basis
\be
|\Psi_{mn}\rangle=\frac{1}{\sqrt d}\sum_{k=0}^{d-1}\omega^{mk}|k+n\rangle\otimes|k\rangle \label{bellbasis}
\ee
where $\omega=\mathrm{e}^\frac{2\pi \mathrm{i}}{d}$.

There are $d^2$ possible results where each one is encoded into a bit string of length $\log d^2$. We label them by $(m,n)$.
 Measurement made by Alice demolishes her state $|\Psi\rangle$ and the only thing she can do is to send a message with information gained by it to Bob. After having received from Alice the outcome of her measurement, Bob performs particular unitary transformation
\be
U_{mn}=\sum_{p=0}^{d-1}\omega^{mp}|n+p\rangle\langle p| \label{unitary}
\ee
on his particle and recreates the state $|\Psi\rangle$. This is a standard procedure which for $d=2$ gives teleportation of qubits \cite{Bennett93}. More general one is presented in \cite{Banaszek00}.

\section{AVERAGE FIDELITY WITH A NOISY CLASSICAL CHANNEL}
Consider a situation in which Alice has to teleport an arbitrary qudit state $|\Psi\rangle$ to Bob. The density matrix od this state is simply $\rho_A=|\Psi\rangle\langle\Psi|$ while Bob's state (not pure in general) after the teleportation we label by $\rho_B$. The efficiency of this process due to the fact that the state of Alice is pure, can be described as the average fidelity calculated over all states $|\Psi\rangle$.
\be
f_{\textrm{avg}}=\int\dpsi\tr(\rho_A\rho_B)=\int\dpsi \langle\Psi|\rho_B|\Psi\rangle \label{averagefidelity}
\ee
where $\dpsi$ is a Haar measure.

When there is no classical communication between the parties the average fidelity is simply $\frac{1}{d}$ whether they share some entangled state or not. It is straightforward to show it by using (\ref{averagefidelity}) and $\rho_B = \frac{1}{d}\mathbb{I}$.

On the other hand, one can assume that Alice and Bob have perfect classical communication and no shared entangled state. In this case Alice measures her state in some basis spanned by vectors $|i\rangle|$ and sends the result (classical information about the basis and the measurement outcome) to Bob who simply prepares the state the system of Alice collapsed to. So with probability $|\langle\Psi|i\rangle|^2$ Bob will produce the state $|i\rangle$ ($\rho_B=\sum_i|\langle\Psi|i\rangle|^2|i\rangle\langle i|$). Then again using (\ref{averagefidelity}) and properties of trace we have
\be \nonumber
f_{\textrm{avg}}&=&\mathrm{tr}\int\dpsi\rho_A\sum_i|\langle\Psi|i\rangle|^2|i\rangle\langle i|\\ \nonumber
&=&\mathrm{tr}\left(\int\dpsi|\Psi\rangle|\Psi\rangle\langle\Psi|\langle\Psi|\sum_i|i\rangle| i\rangle\langle i|\langle i|\right).
\ee
Here, integration under the trace goes over all symmetric elements of $d^2$ dimensional Hilbert space and it commutes with all the operators representing such states. Due to Schur's lemma it is proportional to the symmetric projector $\mathbb{P}_{SYM}$ where the proportionality factor is determined by the trace condition $\tr \int\dpsi|\Psi\rangle|\Psi\rangle\langle\Psi|\langle\Psi|=1$. In this case we have

\be
f_{\textrm{avg}}&=&\mathrm{tr}\left(\frac{\mathbb{P}_{SYM}}{d_{SYM}}\sum_i|i\rangle| i\rangle\langle i|\langle i|\right)=\frac{2}{d+1} \label{quantumlimit}
\ee
where $d_{SYM}=\frac{d(d+1)}{2}$ is a dimension of symmetric space.

This is a well known limit for the teleportation protocols \cite{Horodeccy99} which means that greater values can only be achieved by quantum processes based on entanglement between Alice and Bob.

Finally, we may ask what is the maximal fidelity of the process with the use of a noisy classical channel without any restrictions on shared entanglement? Let us assume, that we are allowed one use of a channel and Bob by using it can read one out of $d^2$ messages built with $\log d^2$ bits. If we assume that Alice sends to Bob a message $(m,n)$, what takes place in the teleportation protocol, and that Bob receives $(a,b)=(n+i,m+j)$ with probability $p_{ij}(mn)$ then the channel can be characterized by error probabilities as
\be
p_{ij}(mn)=p(m+i,n+j|m,n) \label{probabilities}
\ee
where the addition is modulo d and no error event occurs with probability $p_{00}(mn)$. This is the most general description of a channel with $d^2$ inputs and outputs.

After receiving the information Bob performs unitary transformation $U_{ab}$ on his particle. This transformation, by using (\ref{unitary}), can be split into two parts and written as
\be
U_{ab}=\omega^{-mj}U_{ij}U_{mn}. \label{uijumn}
\ee
Operation $U_{mn}$ reproduces the state $\rho_A$ and $U_{ij}$ introduces an error with probability $p_{ij}(mn)$. Finally, after performing his operation and by using
\be
\bar{p}_{ij}=\frac{1}{d^2}\sum_{m,n=0}^{d-1} p_{ij}(mn) \label{averageprobability}
\ee
Bob's state can be described as
\be
\rho_B=\sum_{i,j=0}^{d-1}\bar{p}_{ij}U_{ij}\rho_A U_{ij}^\dag.
\ee
To compute the average fidelity, the formula (\ref{averagefidelity}) has to be used again. Thus we obtain
\be
f_{\textrm{avg}}=\mathrm{tr}\sum_{i,j=0}^{d-1}\bar{p}_{ij}\int\dpsi\rho_A U_{ij}\rho_A U_{i,j}^\dag.
\ee
Because $U_{00}$ is identity we can simplify it as

\be
f_{\textrm{avg}}=\bar{p}_{00}+\mathrm{tr}\sum_{i+j>0}\bar{p}_{ij}\int\dpsi\rho_A U_{ij}\rho_A U_{ij}^\dag \label{fidelityraw}
\ee

When $\bar{p}_{00}=1$ then, of course, we have a perfect teleportation fidelity. To calculate the more general case we use the following lemma.

\begin{lemma} For $\rho_A=|\Psi\rangle\langle\Psi|$ and $U_{ij}$ described by (\ref{unitary}) value of ~$\tr\int\dpsi\rho_A U_{ij}\rho_A U_{ij}^\dag$ for all $i$ and $j$ (excluding $U_{00}$) is equal to $\frac{1}{d+1}$.
\end{lemma}
\begin{proof}
Main tools of this proof are again Schur's lemma and properties of symmetric and antisymmetric projectors. Using them we can write
\be \nonumber
&&\tr\int\dpsi\rho_A U_{ij}\rho_A U_{ij}^\dag=\tr\left(U_{ij}\otimes U_{ij}^\dag\int\dpsi|\Psi\rangle|\Psi\rangle\langle\Psi|\langle\Psi|\right) \\ \nonumber
 &=&\tr\left(U_{ij}\otimes U_{ij}^\dag\frac{\mathbb{P}_{SYM}}{d_{SYM}}\right)=\tr\left(U_{ij}\otimes U_{ij}^\dag\frac{\mathbb{V}-\mathbb{I}}{2d_{SYM}}\right)=\\
 &=&\frac{1}{d(d+1)}\left(\tr(U_{ij} U_{ij}^\dag)+\tr(U_{ij})\tr( U_{ij}^\dag)\right)=\frac{1}{d+1}
\ee
where $\mathbb{V}=\mathbb{P}_{SYM}-\mathbb{P}_{ASYM}$ and for all $i$ and $j$ operators $U_{ij}$ (excluding unity) are traceless.
\end{proof}

 Now it is easy to derive the average fidelity depending on the probability $\bar{p}_{00}$
\be
f_{\textrm{avg}}&=&\frac{d\bar{p}_{00}+1}{d+1}. \label{fidelityresult}
\ee
So to have the fidelity of the teleportation process better or equal to $\frac{2}{d+1}$ the average probability $\bar{p}_{00}$ of sending correct message through noise channel should be greater or equal $\frac{1}{d}$.

Below we present another derivation of the same result which we include because it shows how imperfections of classical channel can be moved to the shared state.

One can write (\ref{uijumn}) differently as
\be
U_{ab}=\omega^{-in}U_{mn}U_{ij}.
\ee
 This situation is equivalent to teleportation protocol with a perfect classical channel in which Bob interferes with his part of quantum source by preforming on it unitary operation $U_{ij}$ with probability $p_{ij}$. This action will change pure maximally entangled state $\rho_{\Phi}=|\Phi\rangle\langle\Phi|$ into a mixed state $\rho_{\Phi}^B=\sum_{ij}p_{ij}\mathbb{I}\otimes U_{ij}\rho_{\Phi} \mathbb{I}\otimes U_{ij}^\dag$.
 Using the previously accepted definitions is easy to show that
 \be
 \rho_{\Phi}^B=p_{00}\rho_{\Phi}+\frac{1}{d}\sum_{i+j>0}p_{ij}\sum_{kl}|k\rangle\langle l|\otimes U_{ij}|k\rangle\langle l| U_{ij}^\dag.
 \ee
To calculate the average fidelity in this case one of the results from  \cite{Horodeccy99} can be used. In that paper authors gave formula $\bar{f}_{max}=(dF_{max}+1)/(d+1)$ where $F_{max}$ is the maximal overlap between the state (built by LOCC) used in the protocol and the state giving a perfect teleportation. Here
\be
F_{max}=\tr\rho_\Phi^B \rho_\Phi.
\ee
 Because operators $U_{ij}$ are traceless we have $F_{max}=p_{00}$ which gives (\ref{fidelityresult}) as well. Of course the result will be the same if instead of Bob Alice will disturb her part of the singlet state (in this case $\rho_{\Phi}^A=\sum_{ij}p_{ij}U_{ij}\otimes\mathbb{I}\rho_{\Phi} U_{ij}^\dag\otimes\mathbb{I} $).

\section{Minimal capacity of the classical channel}
The amount of classical communication is given by the mutual information between Alice's input and Bob's output \cite{Nielsen10}. Capacity then can be written as a function of the dimension of the state being teleported and probabilities (\ref{probabilities})  as
\be\nonumber
C
&=&\log d^2 -\sum_{mn}p_{mn}H(p_{ij}(mn))\\
&=&\log d^2 + \sum_{mn}p_{mn}\sum_{ij}p_{ij}(mn)\log p_{ij}(mn)
 \label{capacity1}
\ee
where $H(p_{ij}(mn))$ is the Shannon entropy and $p_{mn}$ is the probability that Alice will get the result $(m,n)$. We know of course that in a typical scheme all $p_{mn}=\frac{1}{d}$ but we can use more a general distribution if we are interested in minimum of $C$ taken over all errors probabilities. Here we can't use (\ref{averageprobability}) but we can split it into
\be \nonumber
C=\log d^2 &+&\sum_{mn}p_{mn}p_{00}(mn)\log p_{00}(mn)\\
&+& \sum_{mn}p_{mn}\sum_{i+j>0}p_{ij}(mn)\log p_{ij}(mn).
\ee
Because $H(p_{ij}(mn))$ is a convex function $C$ has minimum if all probabilities $p_{ij}$ are independent of information $(m,n)$ sent through classical channel and any error occurs with the same probability. This can be written as
\be
\forall_{mn}~~p_{00}(mn)&=&p_{00}\\
\forall_{mnij;i+j>0}~~p_{ij}(mn)&=&p
\ee
where $p_{00}$ and $p$ are constant values related by
\be
p=\frac{1-p_{00}}{d^2-1}.
\ee
After that it is straightforward to show that the minimal capacity does not depend on $p_{mn}$ and is given by
\be\nonumber
C_{min}(d,p_{00})&=&\log d^2 +p_{00}\log p_{00}\\&+&(1-p_{00})\log \frac{1-p_{00}}{d^2-1}.\label{cmin0}
\ee

Now we can answer the question what is the minimal capacity of classical channel needed to reach the limit (\ref{quantumlimit}). Because, we are using a channel in which $\bar{p}_{00}=p_{00}$ and (\ref{fidelityresult}) it is enough to substitute $p_{00}=\frac{1}{d}$ to (\ref{cmin0}) which gives
\be \nonumber
C_{min}(d,\frac{1}{d})=\log d-(1-\frac{1}{d})\log (d+1). \label{cmin}
\ee
For $d=2$ we reproduce the result ($C_{min}=0.208$) for 2-bit noisy classical channel derived in \cite{Banik12}. What is interesting here is the fact that the function is not monotonic and has a maximum ($C_{min}=0.259$) for $d=4$ (see FIG.\ref{max4}).

\begin{figure}[ht]
\includegraphics[width=0.37\textwidth]{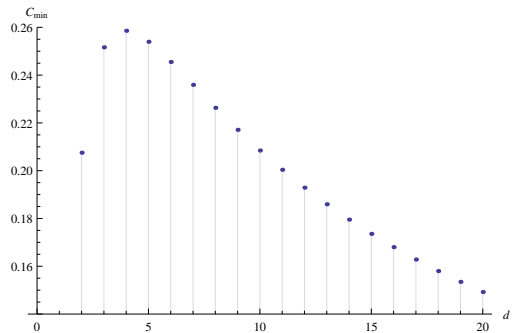}
\caption{Threshold channel capacity for entanglement activation as a function $C_{min}(d,\frac{1}{d})$ of the dimension of the state being teleported.}
\label{max4}
\end{figure}

\section{Comparison of channels}
So far we studied the case of a general noisy channel with $d^2$ inputs and outputs. One might ask if, for a given capacity, it is more beneficial to use this channel or the one with less inputs and outputs but no noise. Now we answer this question. We will call these two channels perfect and noisy respectively.

For the perfect channel we have to find the optimal strategy
for Alice and Bob. It is not a difficult
task. Without loss of generality we can assume that
Alice will send a particular part of her result. For example,
if they are allowed to use 2-bit channel in teleportation of
4-dimensional states Bob after being received information
from Alice has $\frac{1}{4}$ probability to perform correct unitary transformation.

To be more precise assume that Alice's
result was 0110 and she agreed with Bob to always send
two first bits of her measurement outcome. After that
Bob will reject 12 out of 16 capabilities and pick with equal probability one out of four (0100, 0101, 0110, 0111) to perform unitary transformation (\ref{unitary}).

For the generalization it is sufficient to put $p_{00} = \frac{2^C}{
d^2}$ use again (\ref{fidelityraw}) and lemma1. Fidelity for perfect channel can
be than written as
\be
f_{\textrm{avg}}^p(C,d)=\frac{\frac{2^C}{d}+1}{d+1}.
\ee
It's easy to see that for $C=\log d^2$ this fidelity is 1 and that beyond $C=\log d$ it is greater than (\ref{quantumlimit}).

\begin{figure}[ht]
\includegraphics[width=0.37\textwidth]{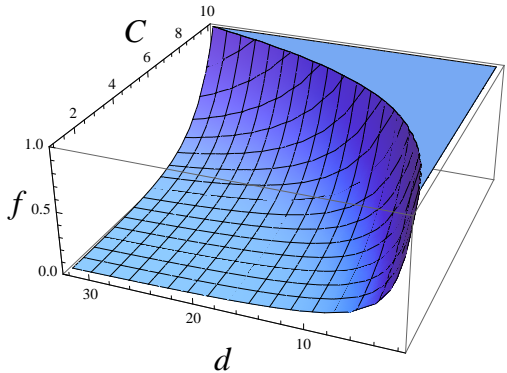} 
\caption{The average fidelity for the perfect channel as a function of the dimension of the teleported state and the capacity of that channel. Here $f=f_{\textrm{avg}}^p(C,d)$, $d$ is the dimension and $C$ stands for the capacity of the channel. }
\label{fprys}
\end{figure}

In the case of the noisy channel finding formula for average fidelity $f_{\textrm{avg}}^n$ is not so simple. There is no algebraic way to do this but numerically it can by easily achieved by using (\ref{cmin0}) and (\ref{fidelityresult}) where average probability $\bar{p}_{00}$ goes to $p_{00}$. 

 \begin{figure}[ht]
\centering
\includegraphics[width=0.37\textwidth]{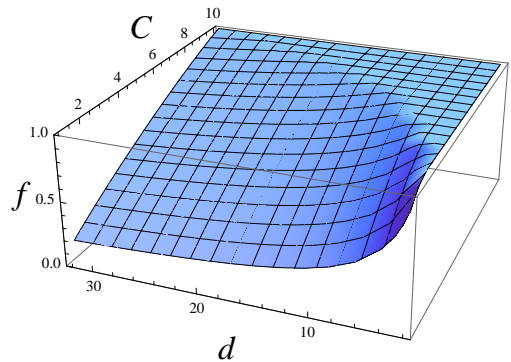}
\caption{The average fidelity for the perfect channel as a function of the dimension of the teleported state and the capacity of that channel.  Here $f=f_{\textrm{avg}}^n(C_{min},d)$, $d$ is the dimension and $C$ stands for the capacity of the channel.}
\label{fnrys}
\end{figure}

Comparing the fidelities shown in FIG.\ref{fprys} and FIG.\ref{fnrys} clearly shows that the teleportation protocol with the same classical amount of information is better if the noisy classical channel is used. The difference between the observed fidelities is shown in FIG.\ref{fn-fprys}.
 
 \begin{figure}[ht]
\includegraphics[width=0.37\textwidth]{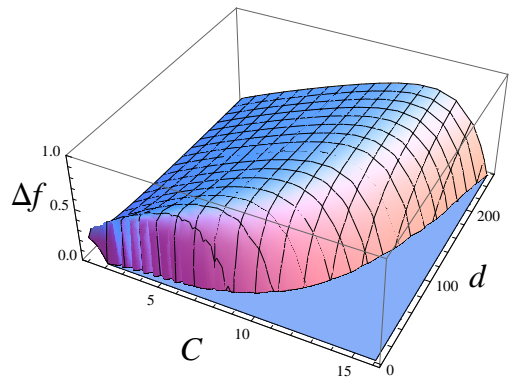} 
\caption{ The difference $f_{\textrm{avg}}^n(C_{min},d)-f_{\textrm{avg}}^p(C,d)$ between the fidelities obtained for the same capacity using different channels.}
\label{fn-fprys}
\end{figure}

\section{conclusions}

We studied the threshold amount of the classical communication required for the teleportation protocol to exceed maximal classical fidelity. We have shown its amount depends on the dimension of the teleported state but is, interestingly, not monotonic and reaches maximum for $d=4$. We have also compared different channels of the same capacity and found that, for teleportation purposes, the one with white noise is optimal. We have restricted ourselves to the standard teleportation protocols involving maximally entangled states and measurements in Bell basis. We conjecture that our results hold also in the general case but proving this is an open avenue of research. It would also be interesting to see how the threshold value of communication changes if some restrictions on shared states are put.

\acknowledgments

We thank Micha\l{} Horodecki for helpful discussions. This work is a part of the Foundation for Polish Science TEAM project cofinanced by the EU European Regional Development Fund. M.P. is supported by ERC grant QOLAPS and U.K. EPSRC. W.L. is supported by the National Centre for Research and Development (Chist-Era Project QUASAR).


\begin{thebibliography}{99}
\bibitem{Bennett93}
C.H. Bennett, G. Brassard, C. Crepeau, R. Jozsa, A. Peres and W.K. Wootters, Phys. Rev. Lett {\bf70} 1895 (1993)
\bibitem{Dur00}
W. Dur, J.I. Cirac, 2000c, J. Mod. Opt.{\bf47} 247
\bibitem{Vaidman94}
L. Vaidman, Phys. Rev.  A{\bf49} 1473-1476 (1994)
\bibitem{Horodeccyreview09}
R. Horodecki, P. Horodecki, M. Horodecki and K. Horodecki,  Rev. Mod. Phys. {\bf81} 865-942 (2009)
\bibitem{Jung09}
E. Jung \emph{et al.} Phys. Rev. A {\bf78},  012312 (2008)
\bibitem{DiFranco12}
C. Di Franco and D. Ballester.  Phys. Lett. A 374.31 (2010): 3164-3169.
\bibitem{Taketami11}
B.G. Taketani, F. de Melo and R. L. de Matos Filho. Phys. Rev. A 85, 020301(R) (2012)
\bibitem{Oh02}
S. Oh, S. Lee and H. Lee. Phys. Rev. A {\bf66},  022316 (2012)
\bibitem{Banaszek00}
K. Banaszek, Phys. Rev.  A{\bf62} 024301 (2000)
\bibitem{Banik12}
M. Banik and R. Gazi, \emph{arXiv:1204.3840v1 [quant-ph]}
\bibitem{Werner98}
R.F. Werner, Phys. Rev.  A{\bf50} 3 (1998)
\bibitem{Horodeccy99}
M. Horodecki, P. Horodecki, R. Horodecki, Phys. Rev.  A{\bf60} 1888 (1999)
\bibitem{Nielsen10}
  M.A. Nielsen and I.L. Chuang, Quantum computation and quantum information. Cambridge university press, 2010.


\end{thebibliography}
\end{document}